\documentclass[a4paper,10pt]{article}
\usepackage{amsthm,amsmath,amssymb}
\usepackage{microtype}
\usepackage{paralist}
\usepackage{charter,eulervm}
\usepackage{subfig}
\usepackage{tabularx}
\usepackage[pdftex,breaklinks,colorlinks,
linkcolor=blue,
citecolor=blue,
urlcolor=blue]{hyperref}
\usepackage{tkz-graph}
\usetikzlibrary{arrows,shapes,decorations.pathmorphing}

\theoremstyle{definition} 
\newtheorem{theorem}{Theorem}[section]
\newtheorem{lemma}[theorem]{Lemma}
\newtheorem{corollary}[theorem]{Corollary}
\newtheorem{proposition}[theorem]{Proposition}
\newtheorem{definition}{Definition}
\newtheorem{kernelrule}{Reduction}
\newtheorem{claim}{Claim}
\newtheorem{remark}[theorem]{Remark}
\def\boxit#1{\vbox{\hrule\hbox{\vrule\kern4pt
      \vbox{\kern1pt#1\kern1pt} \kern2pt\vrule}\hrule}}

\newcommand{\first}[1]{\mathtt{first}(#1)}
\newcommand{\last}[1]{\mathtt{last}(#1)}
\newcommand{\stpath}[2]{$#1$-$#2$ path}
\newcommand{\stsep}[2]{$#1$-$#2$ separator}
\newcommand{\cT}{\ensuremath{{\cal T}}}
\title{Chordal Editing is Fixed-Parameter Tractable\footnote{Supported
    by the European Research Council (ERC) grant 280152 and the
    Hungarian Scientific Research Fund (OTKA) grant NK105645.  A
    preliminary version of this paper appeared in the proceedings of
    STACS 2014.}}

\author{\sc Yixin Cao\thanks{Institute for Computer Science and
    Control, Hungarian Academy of Sciences,
    Email:\href{mailto:yixin@sztaki.hu}{\tt
      yixin@sztaki.hu}, \href{mailto:dmarx@cs.bme.hu}{\tt
      dmarx@cs.bme.hu}.} \addtocounter{footnote}{-1} \and D\'aniel
  Marx\footnotemark}

\begin{document}
\maketitle
\begin{abstract}
  Graph modification problems are typically asked as follows: is there
  a small set of operations that transforms a given graph to have a
  certain property.  The most commonly considered operations include
  vertex deletion, edge deletion, and edge addition; for the same
  property, one can define significantly different versions by
  allowing different operations.  We study a very general graph
  modification problem which allows all three types of operations:
  given a graph $G$ and integers $k_1$, $k_2$, and $k_3$, the
  \textsc{chordal editing} problem asks whether $G$ can be transformed
  into a chordal graph by at most $k_1$ vertex deletions, $k_2$ edge
  deletions, and $k_3$ edge additions.  Clearly, this problem
  generalizes both \textsc{chordal vertex/edge deletion} and
  \textsc{chordal completion} (also known as \textsc{minimum
    fill-in}).  Our main result is an algorithm for \textsc{chordal
    editing} in time $2^{O(k\log k)}\cdot n^{O(1)}$, where
  $k:=k_1+k_2+k_3$ and $n$ is the number of vertices of $G$.
  Therefore, the problem is fixed-parameter tractable parameterized by
  the total number of allowed operations.  Our algorithm is both more
  efficient and conceptually simpler than the previously known
  algorithm for the special case \textsc{chordal deletion}.
 \end{abstract}

\section{Introduction} 
A graph is chordal if it contains no hole, that is, an induced cycle
of at least four vertices.  After more than half century of intensive
investigation, the properties and the recognition of chordal graphs
are well understood.  Their natural structure earns them wide
applications, some of which might not seem to be related to graphs at
first sight.  During the study of Gaussian elimination on sparse
positive definite matrices, Rose
\cite{rose-70-fill-in,rose-72-sparse-matrix} formulated the
\textsc{chordal completion} problem, which asks for the existence of a
set of at most $k$ edges whose insertion makes a graph chordal, and
showed that it is equivalent to \textsc{minimum fill-in}.  Balas and
Yu \cite{balas-86-find-clique-by-maximal-chordal-subgraph} proposed a
heuristics algorithm for the maximum clique problem by first finding a
maximum spanning chordal subgraph (see also
\cite{xue-94-chordel-edge-deletion}).  This is equivalent to the
\textsc{chordal edge deletion} problem, which asks for the existence
of a set of at most $k$ edges whose deletion makes a graph chordal.
Dearing et al.~\cite{dearing-88-maximal-chordal-subgraphs} observed
that a maximum spanning chordal subgraph can also be used to find
maximum independent set and sparse matrix completion.  This
observation turns out to be archetypal: many NP-hard problems
(coloring, maximum clique, etc.) are known to be solvable in
polynomial time when restricted to chordal graphs, and hence admit a
similar heuristics algorithm.  

Cai \cite{cai-03-parameterized-coloring} extended this to the exact
setting.  He studied the coloring problems on graphs close to certain
graph classes.  In particular, he asked the following question: given
a chordal graph $G$ on $n$ vertices with $k$ additional edges (or
vertices), can we find a minimum coloring for $G$ in $f(k)\cdot
n^{O(1)}$ time?  The edge version was resolved by Marx
\cite{marx-param-chordal-full} affirmatively.  His algorithm needs as
part of the input the additional edges; to find them is equivalent to
solving the \textsc{chordal edge deletion} problem.  One may observe
that though with slightly different purpose, the inspiration behind
\cite{balas-86-find-clique-by-maximal-chordal-subgraph,dearing-88-maximal-chordal-subgraphs}
and \cite{cai-03-parameterized-coloring} are exactly the same.

All aforementioned three modification problems , unfortunately but
understandably, are NP-hard
\cite{yannakakis-81-minimum-fill-in,natanzon-01-edge-modification,
  krishnamoorthy-79-node-deletion,lewis-80-node-deletion-np}.
Therefore, early work of Kaplan et
al.~\cite{kaplan-99-chordal-completion} and
Cai~\cite{cai-96-hereditary-graph-modification} focused on their
parameterized complexity, and proved that that the \textsc{chordal
  completion} problem is fixed-parameter tractable.  Recall that a
problem, parameterized by $k$, is {\em fixed-parameter tractable
  (FPT)} parameterized by $k$ if there is an algorithm with runtime
$f(k)\cdot n^{O(1)}$, where $f$ is a computable function depending
only on $k$ \cite{downey-fellows-99}.  Marx
\cite{marx-10-chordal-deletion} showed that the complementary deletion
problems, both edge and vertex versions, are also FPT.  Here we
consider the generalized \textsc{chordal editing} problem that
combines all three operations: can a graph be made chordal by deleting
at most $k_1$ vertices, deleting at most $k_2$ edges, and adding at
most $k_3$ edges.  On the formulation we have two quick remarks.
First, it does not make sense to add new vertices, as chordal graphs
are hereditary (i.e., any induced subgraph of a chordal graph is
chordal).  Second, the budgets for different operations are not
transferable, as otherwise it degenerates to \textsc{chordal vertex
  deletion}.  Our main result establishes the fixed-parameter
tractability of \textsc{chordal editing} parameterized by $k:=k_1 +
k_2 + k_3$.
\begin{theorem}[{\bf Main result}]\label{thm:alg-chordal-editing}
  There is a $2^{{O}(k \log{k})}\cdot n^{O(1)}$-time algorithm for
  deciding, given an $n$-vertex graph $G$, whether there are a set
  $V_-$ of at most $k_1$ vertices, a set $E_-$ of at most $k_2$
  edges, and a set $E_+$ of at most $k_3$ non-edges, such that the
  deletion of $V_-$ and $E_-$ and the addition of $E_+$ make $G$ a
  chordal graph.
\end{theorem}

As a corollary, our algorithm implies the {fixed-parameter
  tractability} of \textsc{chordal edge editing}, which allows both
edge operations but not vertex deletions---we can try every
combination of $k_2$ and $k_3$ where $k_2+k_3$ does not exceed the
given bound---resolving an open problem asked by Mancini
\cite{mancini-08-thesis}.  Moreover, we get a new FPT algorithm for
the special case {\sc chordal deletion}, and it is far simpler and
faster than the algorithm of \cite{marx-10-chordal-deletion}.

\subparagraph{Motivation.}  In the last two decades, graph
modification problems have received intensive attention, and promoted
themselves as an independent line of research in both parameterized
computation and graph theory.  For graphs representing experimental
data, the edge additions and deletions are commonly used to model
false negatives and false positives respectively, while vertex
deletions can be viewed as the detection of outliers.  In this
setting, it is unnatural to consider any single type of errors, while
the \textsc{chordal editing} problem formulated above is able to
encompass both positive and negative errors, as well as outliers.  We
hope that it will trigger further studies on editing problems to
related graph classes, especially interval graphs and unit interval
graphs.

Further, since it is generally acknowledged that the study of chordal
graphs motivated the theory of perfect graphs
\cite{hajnal-58-chordal-graphs, 
  berge-67-some-perfect-graphs}, the importance of chordal graphs
merits such a study from the aspect of structural graph theory.

\subparagraph{Related work.}  Observing that a large hole cannot be
fixed by the insertion of a small number of edges, it is easy to
devise a bounded search tree algorithm for the {\sc chordal
  completion} problem \cite{kaplan-99-chordal-completion,
  cai-96-hereditary-graph-modification}.  No such simple argument
works for the deletion versions: the removal of a single vertex/edge
suffices to break a hole of an arbitrary length.  The way Marx
\cite{marx-10-chordal-deletion} showed that this problem is FPT is to
(1) prove that if the graph contains a large clique, then we can
identify an irrelevant vertex whose deletion does not change the
problem; and (2) observe that if the graph has no large cliques, then
it has bounded treewidth, so the problem can be solved by standard
techniques, such as the application of Courcelle's Theorem.  In
contrast, our algorithm uses simple reductions and structural
properties, which reveal a better understanding of the deletion
problems, and easily extend to the more general \textsc{chordal
  editing} problem.

Of all the vertex deletion problems, we would like to single out
\textsc{feedback vertex set}, \textsc{interval vertex deletion}, and
\textsc{unit interval vertex deletion} for a special comparison.
Their commonality with \textsc{chordal vertex deletion} lies in the
fact that the graph classes defining these problems are proper subsets
of chordal graphs, or equivalently, their forbidden subgraphs contain
all holes as a proper subset.  All these problems admit
single-exponential FPT algorithms of runtime $c^k\cdot n^{O(1)}$,
where the constant $c$ is $3.83$ for \textsc{feedback vertex set}
\cite{cao-10-ufvs}, $10$ for \textsc{interval vertex deletion}
\cite{cao-12-interval-deletion}, and $6$ for \textsc{unit interval
  vertex deletion} \cite{villanger-13-pivd}, respectively.  For these
problems, we can dispose of other forbidden subgraphs (i.e.,
triangles, small witnesses for {asteroidal triples}, and claws) first
and their nonexistence simplifies the graph structure and
significantly decrease the possible configurations on which we conduct
branching (all known algorithms use bounded search trees).
Interestingly, \emph{long holes}, the main difficulty of the current
paper, do not bother us at all in the three algorithms mentioned
above.  This partially explains why a $c^k\cdot n^{O(1)}$-time
algorithm for \textsc{chordal vertex deletion} is so elusive.

\subparagraph{Our techniques.}  As a standard opening step, we use the
{iterative compression} method introduced by Reed et
al.~\cite{reed-04-odd-cycle-transversals} and concentrate on the
compression problem.  Given a solution ($V_-, E_-, E_+$), we can
easily find a set $M$ of at most $|V_-|+ |E_-| + |E_+|$ vertices such
that $G - M$ is chordal.  A clique tree decomposition of $G - M$ will
be extensively employed in the compression step,\footnote{Refer to
  Section~\ref{sec:alg-disjoint-hc} for more intuition behind this
  observation.} where short holes can be broken by simple branching,
and the main technical idea appears in the way we break long holes.
We show that a shortest hole $H$ can be decomposed into a bounded
number of segments, where the internal vertices of each segment, as
well as the part of the graph ``close'' to them behave in a
well-structured and simple way with respect to their interaction with
$M$.  To break $H$, we have to break some of the segments, and the
properties of the segments allow us to show that we need to consider
only a bounded number of canonical separators breaking these
segements. Therefore, we can branch on choosing one of these canonical
separators and break the hole using it, resulting in an FPT algorithm.

\subparagraph{Notation.}  All graphs discussed in this paper shall
always be undirected and simple.  A graph $G$ is given by its vertex
set $V(G)$ and edge set $E(G)$.  We use the customary notation $u\sim
v$ to mean $uv\in E(G)$, and by $v \sim X$ we mean that $v$ is
adjacent to at least one vertex in $X$.  Two vertex sets $X$ and $Y$
are {\em completely connected} if $x \sim y$ for each pair of $x \in
X$ and $y \in Y$.  A hole $H$ has the same number of vertices and
edges, denoted by $|H|$.  We use $N_U(v)$ as a shorthand for $N(v)
\cap U$, regardless of whether $v \in U$ or not; moreover, $N_H(v) :=
N_{V(H)}(v)$ for a hole $H$.  A vertex is \emph{simplicial} if $N[v]$
induce a clique.

A set $S$ of vertices is an \emph{\stsep{x}{y}} if $x$ and $y$ belong
to different components in the subgraph $G -S$; it is \emph{minimal}
if no proper subset of $S$ is an \emph{\stsep{x}{y}}.  Moreover, $S$
is a \emph{minimal separator} if there exists some pair of $x,y$ such
that $S$ is a minimal \stsep{x}{y}.  A graph is chordal if and only if
every minimal separator in it induces a clique
\cite{dirac-61-chordal-graphs}.

Let $\cal T$ be a tree whose vertices, called \emph{bags}, correspond to
the maximal cliques of a graph $G$.  With the customary abuse of
notation, the same symbol $K$ is used for a bag in ${\cal T}$ and its
corresponding maximal clique of $G$.  Let ${\cal T}(x)$ denote the
subgraph of ${\cal T}$ induced by all bags containing $x$.  The tree
$\cal T$ is a \emph{clique tree} of $G$ if for any vertex $x \in
V(G)$, the subgraph ${\cal T}(x)$ is connected.  It is known that the
intersection of any pair of adjacent bags $K$ and $K'$ of \cT\ makes a
minimal separator; in particular, it is a separator for any pair of
vertices $x\in K\setminus K'$ and $y\in K'\setminus K$.  A vertex is
simplicial if and only if it belongs to exactly one maximal clique;
thus, any non-simplicial vertex appears in some minimal separator(s)
\cite{lewis-89-reordering-sparse-matrices}.

In a clique tree $\cT$, there is a unique path between each pair of
bags, and its length is called the \emph{distance} of this pair of
bags; the \emph{distance between two subtrees} is defined to be the
shortest distance between each pair of bags from these two subtrees.
By definition, a pair of vertices $u,v$ of $G$ is adjacent if and only
if $\cT(u)$ and ${\cal T}(v)$ intersect.  Given a pair of nonadjacent
vertices $u$ and $v$, there exists a unique path $\cal P =$($K_u$,
$\dots$, $K_v$) connecting $\cT(u)$ and $\cT(v)$, where $K_u$ and
$K_v$ are the only bags that contain $u$ and $v$ respectively.

\section{Outline of the algorithm}
A subset of vertices is called a \emph{hole cover} of $G$ if its
deletion makes $G$ chordal.  We say that ($V_-, E_-, E_+$), where
$V_-\subseteq V(G)$ and $E_-\subseteq E(G)$ and $E_+\subseteq
V(G)^2\setminus E(G)$, is a \emph{chordal editing set} of $G$ if the
deletion of $V_-$ and $E_-$ and the addition of $E_+$, applied
successively, make $G$ chordal.  Its {\em size} is defined to be the
3-tuple ($|V_-|, |E_-|, |E_+|$), and we say that it is {\em smaller}
than ($k_1,k_2,k_3$) if all of $|V_-|\le k_1$ and $|E_-|\le k_2$ and
$|E_+|\le k_3$ hold true and at least one inequality is strict.  Note
that since chordal graphs are hereditary, it does not make sense to
add new vertices.  The main problem studied in the paper is formally
defined as follows.\medskip

\fbox{\parbox{0.9\linewidth}{
  {\sc chordal editing} ($G, k_1,k_2,k_3$)

\begin{tabularx}{\linewidth}{rX}
  \textit{Input:} & A graph $G$ and three nonnegative integers $k_1$,
  $k_2$, and $k_3$.
  \\
  \textit{Task:} & Either construct a chordal editing set
  $(V_-,E_-,E_+)$ of $G$ that has size at most ($k_1,k_2,k_3$), or
  report that no such a set exists.
\end{tabularx}
}}
\medskip

One might be tempted to define the editing problem by imposing a
combined quota on the total number of operations, i.e., a single
parameter $k=k_1+k_2+k_3$, instead of three separate parameters.
However, this formulation is computationally equivalent to
\textsc{chordal vertex deletion} in a trivial sense, as vertex
deletions are clearly preferable to both edge operations.

We use the technique {\em iterative compression}: we define and solve
a compression version of the problem first and argue that this implies
the fixed-parameter tractability of the original problem.  In the
compression problem a hole cover $M$ of bounded size is given in the
input, making the problem somewhat easier: as $G- M$ is chordal, we
have useful structural information about the graph.  Note that the
definition below has a slightly technical (but standard) additional
condition, i.e., we are not allowed to delete a vertex in $M$.
\medskip

\fbox{\parbox{0.9\linewidth}{
  {\sc chordal editing compression} ($G, k_1,k_2,k_3, M$)

\begin{tabularx}{\linewidth}{rX}
\textit{Input:}& A graph $G$, three nonnegative integers $k_1$, $k_2$, 
and $k_3$, and a hole cover $M$ of $G$ whose size is at most $k_1+ k_2 + 
k_3+1$.
  \\
\textit{Task:} & Either construct a chordal editing
  set $(V_-,E_-,E_+)$ of $G$ such that its size is at most
  ($k_1,k_2,k_3$) and $V_-$ is disjoint from $M$, or report that no
  such a set exists.
\end{tabularx}
}
}
\medskip

The hole cover $M$ is called the \emph{modulator} of this instance.
We use $k:= k_1+k_2+k_3$ to denote the total numbers of operations.
The main part of this paper will be focused on an algorithm for {\sc
  chordal editing compression}.  Its outline is described in
Figure~\ref{fig:alg}.  We will endeavor to prove the following
theorem.

\begin{theorem}\label{thm:alg-compression}
  {\sc chordal editing compression} is solvable in time $2^{{O}(k
    \log{k})}\cdot n^{O(1)}$.
\end{theorem}
\begin{figure}[t]
\centering
\fbox{\parbox{0.7\linewidth}{
    \vspace*{-3mm}    
    \begin{enumerate}
      \small
      \itemindent3mm
      \itemsep-1mm
      \setcounter{enumi}{-1}
    \item {\bf return} if $G$ is chordal or one of
      $k_1$, $k_2$, and $k_3$ becomes negative;
    \item find a shortest hole $H$;
    \item {\bf if} $H$ is shorter than $k + 4$ {\bf then} guess a way to
      fix it; {\bf goto} 0.
    \item {\bf else} decompose $H$ into $O(k^3)$ segments;\\
      \hspace*{7ex} guess a segment and break it;
    \item  {\bf goto} 0.
    \end{enumerate}
  }
}
\caption{Outline of our algorithm for \textsc{chordal editing
    compression}.}
\label{fig:alg}
\end{figure}
Let us briefly explain here steps 1 and 2 of the algorithm for {\sc
  chordal editing compression}, while leaving the main technical part,
step 3, for later sections.  We can find in time ${O}(n^3(n+m))$ a
shortest hole $H$ as follows: we guess three consecutive vertices
$\{v_1,v_2,v_3\}$ of $H$, and then search for the shortest
\stpath{v_1}{v_3} in $G - (N[v_2]\setminus \{v_1,v_3\})$.  In order to
destroy a hole $H$, we need to perform at least one of the possible
$|V(H)\setminus M|$ vertex deletions (vertices in $M$ are avoided
here), $|H|$ edge deletions, or $O(|H|^2)$ edge insertions that affect
$H$.  Therefore, if the length of $H$ is no more than $k + 3$, then we
can fix it easily by branching into $O(k^2)$ direction.  Hence we may
assume $|H| \ge k +4 > k_3 + 3$.  Such a hole cannot be fixed with
edge additions only; thus at least one deletion has to occur on this
hole.  As we shall see in Section~\ref{sec:segments}, the hole can be
divided into a bounded number of ``segments'' (paths), of which at
least one needs to be ``broken.''  In our case, breaking a segment
means more than deleting one vertex or edge from it, and it needs a
strange mixed form of separation: we have to separate two vertices by
removing both edges and vertices. We study this notion of mixed
separation on chordal graphs in Section~\ref{sec:mixed-separators}.
Finally, we show in Section~\ref{sec:alg-chordal-editing} that there
is a bounded number of canonical ways of breaking a segment and we may
branch on choosing one segment and one of the canonical ways of
breaking it. This completes the proof of
Theorem~\ref{thm:alg-compression}, which enables us to prove
Theorem~\ref{thm:alg-chordal-editing}.

\begin{figure*}[t]
\setbox4=\vbox{\hsize28pc \noindent\strut
\begin{quote}
  \vspace*{-5mm} \footnotesize {\bf Algorithm chordal-editing}($G,
  k_1, k_2, k_3$)
  \\
  {Input}: a graph $G$ and three nonnegative integers $k_1$, $k_2$,
  and $k_3$.
  \\
  {Output}: a chordal editing set $(V_-,E_-,E_+)$ of $G$ of size at
  most ($k_1,k_2,k_3$), or ``NO.''
  \\[2ex]
  0 \hspace*{2ex} $i := 0$; $V_- := \emptyset$; $E_- := \emptyset$;
  $E_+ := \emptyset$;
  \\
  1 \hspace*{2ex} {\bf if} $i = n$ {\bf then return} ($V_-, E_-,
  E_+$).
  \\
  2 \hspace*{2ex} $X := V_-\cup \{v_{i+1}\}$ and one endpoint (picked
  arbitrarily) from each edge in $E_-\cup E_+$;
  \\
  3 \hspace*{2ex} {\bf for each} $X_-$ of $X$ of size $\le k_1$ {\bf
    do}
  \\
  3.1 \hspace*{4ex} {\bf call} Theorem~\ref{thm:alg-compression} with
  ($G^{i+1} - X_-, k_1 - |X_-|, k_2, k_3, X \setminus X_-$);
  \hfill{\em $M:=X \setminus X_-$ is the modulator.}
  \\
  3.2 \hspace*{4ex} {\bf if} the answer is ($V'_-, E'_-, E'_+$) {\bf
    then}
  \\
  \hspace*{10ex} ($V_-, E_-, E_+$) := ($V'_-\cup X_-, E'_-, E'_+$);
  \\
  \hspace*{10ex} $i := i+1$; {\bf goto} 1;
  \\
  4 \hspace*{2ex} {\bf return ``NO.''} \hfill{\em no subset $X_-$
    works in step 3.}
\end{quote} \vspace*{-3mm} \strut} $$\boxit{\box4}$$
\vspace*{-8mm}
\caption{Algorithm for \textsc{chordal editing}.}
\label{fig:main-alg}
\end{figure*}
\begin{proof}[Proof of Theorem~\ref{thm:alg-chordal-editing}]
  Let $v_1,\dots, v_n$ be an arbitrary ordering of $V(G)$, and let
  $G^i$ be the subgraph induced by the first $i$ vertices.  Note that
  $G^n = G$.  The algorithm described in Figure~\ref{fig:main-alg}
  iteratively finds a chordal editing set of $G^i$ from $i=1$ to $n$;
  the solution for $G^i$ is used in solving $G^{i+1}$.  The algorithm
  maintains as an invariant that ($V_-, E_-, E_+$) is a chordal
  editing set of size at most ($k_1,k_2,k_3$) of $G^i$ for the current
  $i$.  For each $G^i$, note that $|X|\le k+1$, step 3 generates at
  most $2^{O(k)}$ instances of {\sc chordal editing compression}, each
  with parameter at most ($k_1,k_2,k_3$), and thus can be solved in
  $2^{{O}(k \log{k})}\cdot n^{O(1)}$ time.  There are $n$ iterations,
  and the total runtime of the algorithm is thus $2^{{O}(k
    \log{k})}\cdot n^{O(1)}$.
\end{proof}

\section{Segments}\label{sec:segments}
We need to define a hierarchy of vertex sets $V_0, V_1$, and $V_2$.
Each set is a subset of the preceding one, and all of them induce
chordal subgraphs.  Let $A$ denote the set of common neighbors of the
shortest hole $H$ found in step 1 (Figure \ref{fig:alg}), and define
$A_M = A \cap M$ and $A_0 = A \setminus M$.  We can assume that $A$
induces a clique: if two vertices $x,y\in A$ are nonadjacent, then
together with two nonadjacent vertices $v_1$ and $v_3$ of $H$, they
form a 4-hole $x v_1 y v_3 x$.  The following observation follows from
the fact that $H$ is the shortest hole of $G$.
\begin{proposition}\label{lem:at-most-3}
  A vertex not in $A$ is adjacent to at most three vertices of $H$ and
  these vertices have to be consecutive in $H$.
\end{proposition}

The first set is defined by $V_0 = V(G) \setminus (M\cup A)$, and let
$G_0 = G[V_0]$.  Note that $\{M$, $A_0$, $V_0\}$ partitions $V(G)$,
and $H$ is disjoint from $A_0$.  Since $|H|\ge k + 4 > |M|$ and $G_0$
is chordal, the hole $H$ intersects both $M$ and $V_0$.  Every
component of $H-M$ is an induced path of $G_0$, and there are at most
$|M|$ such paths.  We divide each of these paths into $O(k^2)$ parts;
observing $|M|=O(k)$, this leads to to a decomposition of $H$ into
$O(k^3)$ segments.  Let $P$ denote such a path $v_1 v_2 \dots v_{p}$
in $H$, where $v_i\in V_0$ for $1\le i\le p$ and the other neighbors
of $v_1$ and $v_p$ in $H$ (different from $v_2$ and $v_{p-1}$
respectively) are in $M$.  We restrict our attention to paths with
$p>3$ (as there is a trivial bound for shorter paths). For such paths,
Proposition~\ref{lem:at-most-3} implies that the distance between
$v_1$ and $v_p$ in $G_0$ is at least $3$.  A further consequence is
$v_1 \not\sim v_{p}$.

Let us fix a clique tree $\cal T$ for the chordal subgraph $G_0$.  We
take the unique path $\cal P$ of bags $K_1$, $\dots$, $K_{q}$ that
connects the disjoint subtrees $\cT(v_1)$ and $\cT(v_p)$ in \cT, where
$K_1\in \cT(v_1)$ and $K_q\in \cT(v_p)$.  The condition $p>3$ implies
that $q>2$.  The removal of $K_1$ and $K_q$ will separate $\cT$ into a
set of subtrees, one of which contains all $K_{\ell}$ with $1<\ell<
q$; let ${\cal T}_1$ denote this nonempty subtree.  The second set,
$V_1$, is defined to be the union of all bags in ${\cal T}_1$ and
$\{v_1, v_p\}$.  By definition and observing that $V_1$ fully contains
$P$, it induces a connected subgraph.

We then focus on bags in $\cal P$ and their union.  (One may have
judiciously observed that these vertices induce an interval graph.)
From the definition of clique tree, we can infer that $v_1$ and $v_p$
appear only in $K_1$ and $K_q$ respectively, while every internal
vertex of $P$ appears in more than one bags of $\cal P$.  For every
$i$ with $1\le i\le p$, we denote by ${\mathtt{first}(i)}$ (resp.,
${\mathtt{last}(i)}$) the smallest (resp., largest) index $\ell$ such
that $1\le \ell\le q$ and $v_i\in K_\ell$, e.g., ${\mathtt{first}(1)}
= {\mathtt{last}(1)} = {\mathtt{first}(2)} = 1$ and
${\mathtt{last}(p-1)} = {\mathtt{first}(p)} = {\mathtt{last}(p)} = q$.
As $P$ is an induced path, for each $i$ with $1 < i < p$, we have
\begin{equation}
{\mathtt{first}(i)} \le {\mathtt{last}(i-1)} < {\mathtt{first}(i + 1)}
\le {\mathtt{last}(i)}.
\end{equation}
For $1\le\ell<q$, we define $S_{\ell} = K_{\ell}\cap K_{\ell+1}$.  For
any pair of nonadjacent vertices $v_i, v_j$ in $P$, (i.e., $1\le
i<i+1<j\le p$,) all minimal \stsep{v_i}{v_j}s are then
$\{S_{\ell}\mid{\mathtt{last}(i)}\le \ell<{\mathtt{first}(j)}\}$.

The third set, $V_2$, is defined to be the union of vertices in all
induced \stpath{v_1}{v_p}s in $G_0$.  Note that $V_2$ and $A_0$ are
completely connected: given a pair of nonadjacent vertices $x\in V_2$
and $y\in A_0$, we can find a hole of $G - M$ that consists of $y$ and
part of a \stpath{v_1}{v_p} through $x$ in $G_0$.  Since a vertex $x$
is an internal vertex of an induced \stpath{v_1}{v_p} of $G_0$ if and
only if it is in some minimal \stsep{v_1}{v_p} of $G_0$, we have
(noting $q>2$)
\begin{proposition}\label{lem:v-2}
  A vertex is in $V_2\setminus \{v_1,v_p\}$ if and only if it appears
  in more than one bags of $\cal P$.  Moreover, $V_2\setminus
  \{v_1,v_p\} \subseteq \bigcup_{1< \ell<q} K_{\ell}$.
\end{proposition}

The definitions of $V_0$ and $G_0$ depend upon the hole $H$, while the
definitions of $V_1$ and $V_2$ depend upon both the hole $H$ and the
path $P$.  In this paper, the hole $H$ will be fixed, and we are
always concerned with a particular path of $H$, which will be
specified before the usage of $V_1$ and $V_2$.

The set $V_0\setminus V_1$ is easily understood, and we now consider
$V_1\setminus V_2$.  Given a pair of nonadjacent vertices $x,y\in
V_2$, we say that $x$ lies to the \emph{left} (resp., \emph{right}) of
$y$ if the bags of $\cal P$ containing $x$ have smaller (resp.,
greater) indices than those containing $y$.  If an induced path of
$G[V_2]$ consists of three or more vertices, then its endvertices are
nonadjacent and have a left-right relation.  This relation can be
extended to all pairs of consecutive (and adjacent) vertices $x,y$ in
this path, the one with smaller distance to the left endvertex of the
path is said \emph{to the left of the other}.  It is easy to verify
that these two definitions are compatible.
\begin{lemma}\label{lem:branch}
  For any component $C$ of the subgraph induced by $V_1\setminus V_2$,
  the set $N_{V_0}(C)$ induces a clique and there exists $\ell$ such
  that $1< \ell <q$ and $N_{V_0}(C)\subseteq K_{\ell}$.
\end{lemma}
\begin{proof}
  Consider a vertex $x\in C$, which is different from $v_1$ and $v_p$.
  Since $x\in V_1$, it appears in some bag of $\cT_1$.  Recall that
  the only bag of $\cT_1$ that is adjacent to $K_1$ is $K_2$.  We
  argue first that $x\not\in K_1$: recall that $V_1$ is disjoint from
  $K_1\setminus (\{x\}\cup K_2)$, and thus if $x\in K_1$ then it has
  to be in $K_2$ as well, but then $x\in V_2$
  (Proposition~\ref{lem:v-2}), contradicting that $C\subseteq
  V_1\setminus V_2$.  For the same reason, $x\not\in K_q$.  As a
  result, $N_{V_0}(x)\subseteq V_1$, and then $N_{V_0}(C)\subseteq
  V_2$.  It now suffices to show that $N_{V_0}(C)$ induces a clique.
  Suppose for contradiction that there is a pair of nonadjacent
  vertices $x,y\in N_{V_0}(C)$.  We can find an induced
  \stpath{v_1}{v_p} $P'$ through $x$ and $y$; without loss of
  generality, let $x$ lie to the left of $y$, i.e., $P'=v_1 \cdots x
  \cdots y \cdots v_p$.  Let $x'$ and $y'$ be the first and last
  vertices in $P'$ that are adjacent to $C$, and let $x' P'' y'$ be an
  induced path with all internal vertices from $C$.  Note that $x'$
  either is $x$ or lies to the left of $x$ in $P'$ and $y'$ either is
  $y$ or lies to the right of $y$, which imply $x'\not\sim y'$.  Thus
  $v_1 \cdots x' P'' y' \cdots v_p$ is an induced \stpath{v_1}{v_p}
  through $C$, which is impossible.  This completes the proof.
\end{proof}

Such a component $C$ is called a \emph{branch} of $P$, and we say that
it is \emph{near to} $v_i\in P$ if there is an $\ell$ with
$\first{i}\le \ell\le \last{i}$ satisfying the condition of
Lemma~\ref{lem:branch}.  In other words, $C$ is near to $v_i\in P$ if
and only if $N_{V_0}(C)\subseteq N[v_i]$.  Applying
Proposition~\ref{lem:at-most-3} on any vertex in $N_{V_0}(C)$, we
conclude that a branch is near to at most three vertices of $P$.  If
there exists some hole passing through $C$, then $C$ has to be
adjacent to $M$: by Lemma~\ref{lem:branch} and recalling that $V_2$
and $A_0$ are completely connected, $N_{V_0}(C)\cup A_0$ is a clique,
and thus a hole cannot enter and leave $C$ both via $N_{V_0}(C)\cup
A_0$.  The converse is not necessarily true: some branch that is
adjacent to $M$ might still be disjoint from all holes, e.g., if
$N(C)$ is a clique.  This observation inspires us to generalize the
definition of simplicial vertices to sets of vertices.

\begin{definition}
  A set $X$ of vertices is called \emph{simplicial in a graph $G$} if
  $N[X]$ induces a chordal subgraph of $G$ and $N(X)$ induces a clique
  of $G$.
\end{definition}

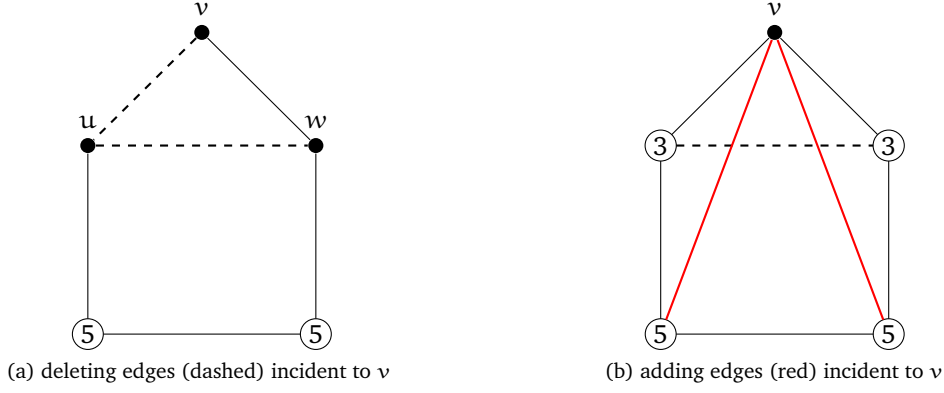
\begin{figure*}[t]
  \centering
  \subfloat[deleting edges (dashed) incident to $v$]{\label{fig:bad-delete}
    \makebox[.45\textwidth]{
  \begin{tikzpicture}
    \tikzstyle{vertex} = [fill=black,circle,inner sep=2pt]
    \tikzstyle{clique} = [draw,circle,outer sep=0,inner sep=1,minimum size=5]
    
    \node [vertex,label=above:$v$] (v2) at (-1.5,3.5) {};
    \node [vertex,label=above:$u$] (v1) at (-3,2) {};
    \node [vertex,label=above:$w$] (v3) at (0,2) {};
    \node [clique] (v4) at (0,-0.5) {$5$};
    \node [clique] (v5) at (-3,-0.5) {$5$};
    \draw (v2) -- (v3) -- (v4) -- (v5) -- (v1);
    \draw [thick, dashed] (v2) -- (v1) -- (v3);
  \end{tikzpicture}
  }}
  \subfloat[adding edges (red) incident to $v$]{\label{fig:bad-add}
    \makebox[.45\textwidth]{\centering
  \begin{tikzpicture}
    \tikzstyle{vertex} = [fill=black,circle,inner sep=2pt]
    \tikzstyle{clique} = [draw,circle,outer sep=0,inner sep=1,minimum size=5]
    
    \node [vertex,label=above:$v$] (v2) at (-1.5,3.5) {};
    \node [clique] (v1) at (-3,2) {$3$};
    \node [clique] (v3) at (0,2) {$3$};
    \node [clique] (v4) at (0,-0.5) {$5$};
    \node [clique] (v5) at (-3,-0.5) {$5$};
    \draw (v1) -- (v2) -- (v3) -- (v4) -- (v5) -- (v1);
    \draw [thick, dashed] (v1) -- (v3);
    \draw [red, thick] (v5) -- (v2) -- (v4);
  \end{tikzpicture}}
  }
  \caption{Possible modifications to a simplicial vertex $v$.
    (\textcircled{$x$} means a clique of $x$ vertices and an edge means all the edges between the two cliques/vertices.) (a) A minimal solution with two edge deletions. (b) A minimal solution with one edge deletion and two edge addition.) }
  \label{fig:bad}
\end{figure*}
It is easy to verify that a simplicial set of vertices is disjoint
from all holes.  This suggests that simplicial sets are irrelevant to
{\sc chordal editing} problem and we may never want to add/delete
edges incident to a vertex in a simplicial set.  However, this is not
true: as Figure~\ref{fig:bad} shows, we may need to add/delete such
edges if $N(X)$ was modified.  As characterized by the following
lemma, this is the only reason for touching it in the solution.  In
other words, a simplicial set $X$ will only concern us after $N(X)$
has been changed.  We say that a chordal editing set ($V_-, E_-, E_+$)
{\em edits} a set $X\subset V(G)$ of vertices if either $V_-$ contains
a vertex of $X$ or $E_-\cup E_+$ contains an edge with at least one
endpoint in $X$.  We use a classic result of Dirac
\cite{dirac-61-chordal-graphs} stating that the graph obtained by
identifying two cliques of the same size from two chordal graphs is
also chordal.
\begin{lemma}\label{lem:chordal-protrusion}
  A minimal chordal editing set edits a simplicial set $U$ only if it
  removes at least one edge induced by $N(U)$. 
\end{lemma}
\begin{proof}
  Let ($V_-, E_-, E_+$) be a minimal editing set of $G$ such that
  $E_-$ does not contain any edge induced by $N(U)$.  We restrict the editing set
  to the subgraph $G - U$, i.e., we consider the set ($V_-\setminus U,
  E_-\setminus (U\times V(G)), E_+\setminus (U\times V(G))$), and let
  $G'$ be the graph obtained by applying it to $G$.  Clearly $G' - U =
  G - U$ is chordal, where $N(U)\setminus V_-$ induces a clique.  Also
  chordal is the subgraph of $G'$ induced by $N[U]\setminus V_-$.  Both
  of them contain the clique $N(U)\setminus V_-$.  Since $G'$ can be
  obtained from them by identifying $N(U)\setminus V_-$, it is also
  chordal.  Then by the minimality of ($V_-, E_-, E_+$), it must be
  the same as ($V_-\setminus U, E_-\setminus (U\times V(G)),
  E_+\setminus (U\times V(G))$), and this proves this lemma.
\end{proof}

Now we are ready to define segments of the path $P$, which are
delimited by some special vertices called {junctions}.  By definition,
a branch is simplicial in $G_0$, but not necessarily simplicial in $G$.  

\begin{definition}[Segment]
  A vertex $v\in P$ is called a \emph{junction (of $P$)} if
  \begin{enumerate}[\itshape(1)]
  \item some bag $K$ that contains $v$ is adjacent to $M\setminus
    A_M$;
  \item some branch near to $v$ is adjacent to $M\setminus A_M$;
  \item some branch near to $v$ is not simplicial in $G$; or
  \item $N_{V_2}(v)$ is not completely connected to $A$.
  \end{enumerate}
  A sub-path $v_s\cdots v_t$ of $P$ is called a \emph{segment},
  denoted by $[v_s, v_t]$, if $v_s$ and $v_t$ are the only junctions
  in it.
\end{definition}

We point out that the four types are not exclusive, and one junction
might be in more than one types.  For a junction $v$ of type (1) or
(2), we say that the vertex in $M\setminus A_M$ used in its definition
\emph{witnesses} it.  Let us briefly explain the intuition behind the
definition of junctions and segments.

\begin{remark}
  For a junction $v$ of type (1) or (2), there is a connection from
  $v$ to $M\setminus A_M$ that is \emph{local} to $v$ in some sense;
  for a junction $v$ of type (3) or (4), there is a hole near to $v$,
  and its disposal might interfere with that of $H$.  On the other
  hand, since there is no junctions inside a segment $[v_s, v_t]$, if
  another hole $H'$ intersects it, then $H'$ has to ``go through the
  whole segment.''  Or precisely, $H'$ necessarily enters and exits
  the segment via $N[v_{s}]$ and $N[v_{t}]$, respectively.
\end{remark}

The definition of junction and segment extends to all paths of $H- M$.
In polynomial time, we can construct $V_0$ for $H$ and $V_1, V_2$ for
each path $P$ of $H- M$, from which all junctions of $H$ can be
identified.  In particular, the endvertices of $P$ are adjacent to
$M\setminus A_M$, hence junctions (of type (1)).  As a result, every
vertex in $V(H)\setminus M$ is contained in some segment, and in each
path of $H-M$, the number of segments is the number of junctions minus
one.

We are now ready for the main result of this section that gives a
cubic bound on the number of segments of $H$.  It should be noted the
constants---both the exponent and the coefficient---in the following
statement are not tight, and the current values simplify the argument
significantly.  Recall that a vertex not in $A$ sees at most three
vertices in $H$, and they have to be consecutive.
\begin{theorem}\label{lem:bound-segments}
  If $H$ contains more than $|M|\cdot (12 k^2 + 87 k + 75)$
  segments, then we can either find a vertex that has to be in $V_-$,
  or return ``NO.''
\end{theorem}
\begin{proof}
  We show that $H$ contains at most $|M|\cdot (12 k^2 + 87 k + 75)$
  junctions.  Recall that there are at most $|M|$ paths in $H - M$. To
  obtain a contradiction, we suppose that some path $P$ of $H - M$
  contains $12 k^2 + 87 k + 75$ junctions.  Let us first attend
  to junctions of type (1) in $P$.
  \begin{claim}\label{claim-1}
    Each $w\in M\setminus A_M$ witness at most $14$ junctions of type
    (1) in $P$.
  \end{claim}
  \begin{proof}
    We are proving a stronger statement of this claim, i.e., $w$
    witness at most $14$ junctions of type (1) in the entire hole $H$.
    Suppose, for contradiction, that $15$ vertices in $H$ appear in
    some bag adjacent to $w$; let $X$ be this set of vertices.
    Assume first that $X$ is consecutive.  At most 3 of them are
    adjacent to $w$, and they are consecutive in $H$.  Thus, we can
    always pick 6 consecutive vertices from $X$ that are disjoint from
    $N_H(w)$; let them be $\{v_i,\dots,v_{i+5}\}$.  By definition,
    there are two vertices $u_1, u_2\in V_0\cap N(w)$ such that
    $u_1\sim v_i$ and $u_2\sim v_{i+5}$.  It is easy to verify that
    $u_2\not\sim v_{i+2}$ and $u_1\not\sim v_{i+3}$ and $u_1\not\sim
    u_2$.  Therefore, we can find an induced \stpath{u_1}{u_2} with
    all interval vertices from $\{v_i,\dots, v_{i+5}\}$.  The length
    of this path is at least $3$, and hence it makes a hole with $w$
    of length at most $9$.  Assume now that $X$ is not consecutive in
    $P$, then we can pick a pair of nonadjacent vertices $v_i,v_j$
    from $X$ such that the $v_\ell\not\in X$ for every $i<\ell<j$.
    There are two vertices $u_1, u_2\in V_0\cap N(w)$ such that
    $u_1\sim v_i$ and $u_2\sim v_j$.  It is easy to verify that $w u_1
    v_i \cdots v_j u_2 w$ is a hole.  By assumption that $|X| \ge 15$,
    we have $j-i\le |H| - 13$.  In either case, we end with a hole
    strictly shorter than $H$.  The contradictions prove this claim.
    \renewcommand{\qedsymbol}{$\lrcorner$}
  \end{proof}
  \begin{claim}\label{claim-2}
    If some vertex $w\in M\setminus A_M$ witnesses $5k + 75$ junctions
    of types (1) and (2) in $P$, then we can return ``NO.''
  \end{claim}
  \begin{proof}
    Let $X$ be this set of junctions, we order them according to their
    indices in $P$ and group each consecutive five from the beginning.
    We omit groups that contain junctions of type (1) witnessed by
    $w$, and in each remaining group, we pair the second and last
    vertices in it.  According to Claim~\ref{claim-1}, we end with at
    least $k+1$ pairs, which we denote by ($v_{\ell_1}, v_{r_1}$),
    $\cdots$, ($v_{\ell_{k+1}}, v_{r_{k+1}}$), $\cdots$.

    For each pair ($v_{\ell_j}, v_{r_j}$), where $1\le j\le k+1$, we
    construct a hole $H_j$ as follows.  By definition, there is a
    branch $C_{\ell_j}$ (resp., $C_{r_j}$) whose neighborhood in $H$
    is a proper subset of
    $\{v_{{\ell_j}-1},v_{\ell_j},v_{{\ell_j}+1}\}$ (resp.,
    $\{v_{{r_j}-1},v_{r_j},v_{{r_j}+1}\}$).  By the selection of the
    pair $v_{\ell_j}$ and $v_{r_j}$ (two vertices of $X$ have been
    skipped in between), they are nonadjacent, and $r_j - \ell_j > 2$.
    Therefore, $C_{\ell_j}$ and $C_{r_j}$ are distinct and necessarily
    nonadjacent.  Since $C_{\ell_j}$ induces a connected subgraph and
    is adjacent to both $w$ and
    $\{v_{{\ell_j}-1},v_{\ell_j},v_{{\ell_j}+1}\}$, we can find an
    induced \stpath{w}{v_{{\ell_j}+1}} $P_{\ell_j}$ with all internal
    vertices from $C_{\ell_j}\cup \{v_{{\ell_j}-1},v_{\ell_j}\}$.
    Likewise, we can obtain an induced \stpath{w}{v_{r_j - 1}}
    $P_{r_j}$ with all internal vertices from $C_{r_j-1}\cup
    \{v_{{r_j}},v_{{r_j}+1}\}$.  These two paths $P_{\ell_j}$ and
    $P_{r_j}$, together with $v_{\ell_j +1}\dots v_{r_j-1}$, make the
    hole $H_{j}$: we have ${\ell_j +1}<{r_j-1}$; for each ${\ell_j
      +1}\le s\le {r_j-1}$, $v_s\not\sim w$; and for each ${\ell_j
      +1}< s< {r_j-1}$, $v_s\not\sim C_{\ell_j}, C_{r_j}$.  This hole
    goes through $w$.  This way we can construct $k+1$ holes, and it
    can be easily verified that they intersect only in $w$.  Since we
    are not allowed to delete $w$, we cannot fix all these holes by at
    most $k$ operations.  Thus we can return ``NO.''
    \renewcommand{\qedsymbol}{$\lrcorner$}
  \end{proof}

  If Claim~\ref{claim-2} applies, then we are already done; otherwise,
  there are at most $|M|\cdot (5k + 74)$ junctions of the first two
  types in $P$.  We proceed by considering the set $B$ of junctions
  that are only of type (3) or (4) but not of the first two types.
  Its number is at least (noting $|M| \le k+1$)
  $$(12 k^2 + 87 k + 75) - (5k+74)\cdot |M| \ge 7 k^2 + 7k + 1.
  $$
  We order $B$ according to their indices in $P$, and let $b_{i}$
  denote the index of the $i$th vertex of $B$ in $P$.  For each $0\le
  i\le k(k+1)$, we use the ($7i+3$)th vertex of $B$ to construct a
  hole $H_i$. Then we argue that this collection of holes either
  allows us to identify a vertex that has to be in the solution, or
  conclude infeasibility.

  The first case is when $v_{b_{7i+3}}$ is of type (4): there is a
  pair of nonadjacent vertices $x\in N_{V_2}(v_{b_{7 i + 3}})$ and
  $y\in A$.  In this case we can assume that $x$ is adjacent to
  neither $v_{b_{7 i + 1}}$ nor $v_{b_{7 i + 5}}$; otherwise $x
  v_{b_{7 i + 1}} y v_{b_{7 i + 3}} x$ or $x v_{b_{7 i + 3}} y v_{b_{7
      i + 5}} x$ is a $4$-hole, which contradicts the fact that $H$ is
  the shortest.  In other words, $x$ only appears in some bag between
  $K_{\last{b_{7 i + 1}}}$ and $K_{\first{b_{7 i + 5}}}$; on the other
  hand, by definition of $V_2$, it appears in at least two of these
  bags.  There is thus an induced \stpath{v_{b_{7 i + 1}}}{v_{b_{7 i +
        5}}} $P_i$ via $x$ in $G[V_2]$.  Starting from $x$, we
  traverse $P_i$ to the left until the first vertex $x_1$ that is
  adjacent to $y$; the existence of such a vertex is ensured by the
  fact that $y\sim v_{b_{7 i + 1}}$.  Similarly, we find the first
  neighbor $x_2$ of $y$ in $P_i$ to the right of $x$.  Then the
  sub-path of $P_i$ between $x_1$ and $x_2$, together with $y$, gives
  the hole $H_i$.  By construction, no vertex of $H_i - y$ is adjacent
  to $v_{b_{7 i}}$ or $v_{b_{7 i + 6}}$.

  In the other case, $v_{b_{7i+3}}$ is type (4): some branch $C_i$
  near to $v_{b_{7i+3}}$ is not simplicial in $G$.  By definition,
  either the subgraph induced by $N(C_i)$ is not a clique, or the
  subgraph induced by $N[C_i]$ is not chordal.  Since $v_{b_{7i+3}}$
  does not satisfy the conditions of type (1) and (2), $N(C_i)\cap
  M\subseteq A_M$, i.e., $N(C_i)\setminus V_0 \subseteq A$.  On the
  other hand, according to Lemma~\ref{lem:branch}, $N(C_i)\cap V_0$
  induces a clique.  Therefore, there must be a pair of nonadjacent
  vertices $x\in N(C_i)\cap V_0$ and $y\in A_M$.  As $C_i$ is near to
  $v_{b_{7i+3}}$, it must hold that $x\in N(v_{b_{7i+3}})$; this has
  already been discussed in the previous case.  Suppose now that
  $N(C_i)$ induces a clique and there is a hole $H_i$ in $N[C_i]$.  We
  have seen that $N[C_i]\cap M = A_M$, thus this hole $H_i$ intersects
  $A_M$; let $w$ be a vertex in $V(H_i)\cap A_M$.  If $H_i$ is
  disjoint from $A_0$, then no vertex in $H_i\setminus M$ can be
  adjacent to $v_{b_{7 i}}$ or $v_{b_{7 i + 5}}$.  Otherwise, it
  contains some vertex $u\in A_0$; noting that $A$ induces a clique,
  $H_i\cap A =\{u, w\}$.  Moreover, $N(C_i)\cap V_2$ is in the
  neighborhood of $v_{b_{7i+3}}$ and therefore $N(C_i)\cap V_2$ and
  $N(C_j)\cap V_2$ are disjoint for $i\neq j$: the existence of a
  vertex $x\in V_2$ adjacent to both $C_i$ and $C_j$ would contradict
  Proposition~\ref{lem:at-most-3} (noting that the distance of
  $v_{b_{7i+3}}$ and $v_{b_{7j+3}}$ is greater than 2 on the hole
  $H$).

    In sum, we have a set $\cal H$ of at least $k(k+1)+1$ distinct
    holes such that
    \begin{inparaenum}[(1)]
    \item each hole in $\cal H$ contains at most one vertex of
      $A_0$, and
    \item the intersection of any pair of them is in $A$.
    \end{inparaenum}
    Recall that each hole has length at least $k + 4$, hence cannot be
    fixed by edge additions only.  If there is a $u\in A_0$ contained
    in at least $k+1$ holes of $\cal H$, then we have to put $u$ into
    $V_-$; otherwise we have to delete distinct elements (edges or
    vertices) to break different holes, which is impossible.  Now
    assume that no such a vertex $u$ exists, then there must be $k+1$
    holes that intersect only in $M$, which allow us to return ``NO.''
\end{proof}

\section{Mixed separators in chordal graphs}
\label{sec:mixed-separators}

Given a pair of nonadjacent vertices $x,y$ of a graph, we say that a
pair of vertex set $V_S$ and edge set $E_S$ is a \emph{mixed
  \stsep{x}{y}} if the deletion of $V_S$ and $E_S$ leaves $x$ and $y$
in two different components; its size is defined to be
($|V_S|, |E_S|$).  A {mixed \stsep{x}{y}} is \emph{inclusion-wise
  minimal} if there exists no other {mixed \stsep{x}{y}} ($V'_S,
E'_S$) such that $V'_S\subseteq V_S$ and $E'_S\subseteq E_S$ and at
least one containment is proper.  
If ($V_S, E_S$) is an
{inclusion-wise minimal} {mixed \stsep{x}{y}} in graph $F$, then each
component of $F - V_S - E_S$ is an induced subgraph of $F$.
Therefore, we have the following characterization of {inclusion-wise
  minimal} mixed separators in chordal graphs.

\begin{proposition}\label{lem:stub-is-chordal}
  In a chordal graph, all components obtained by deleting an
  inclusion-wise minimal \stsep{x}{y} are chordal.
\end{proposition}

Consider an inclusion-wise minimal \stsep{x}{y} ($V_S, E_S$) in a
chordal graph $F$.  Let $\cT^F$ be a clique tree of $F$.  The
degenerated case where $E_S = \emptyset$ is well understood: $V_S$
itself makes an \stsep{x}{y}.  If $E_S\ne \emptyset$, then in the path
that connects $\cT^F(x)$ and $\cT^F(y)$, at least one bag is
disconnected by the deletion of $V_S$ and $E_S$.  This bag contains at
most $|V_S| + |E_S| +1$ vertices.  On the other hand, the remaining
vertices of every bag $K$, i.e., $K\setminus V_S$, appear in either
one or two components of $F - V_S - E_S$.  In the latter
case, the two components are precisely that contain $x$ and
$y$, respectively; otherwise the mixed separator cannot be
inclusion-wise minimal.

\begin{lemma}\label{lem:alg-mix-separator}
  Let $x$ and $y$ be a pair of nonadjacent vertices in a chordal graph
  $F$.  For any pair of nonnegative integers ($a,b$), we can find a
  mixed \stsep{x}{y} of size at most ($a,b$) or asserts its
  nonexistence in time $3^{a+b+1} \cdot |V(F)|^{O(1)}$.
\end{lemma}
\begin{figure*}[ht]
\setbox4=\vbox{\hsize28pc \noindent\strut
\begin{quote}
  \vspace*{-5mm} \footnotesize {\bf Algorithm mixed-separator}($F,
  x,y, a, b$)
  \\
  {\sc input}: a chordal graph $F$, nonadjacent vertices $x$ and $y$,
  and nonnegative integers $a$ and $b$.
  \\
  {\sc output}: a mixed \stsep{x}{y} ($V_S, E_S$) of size at most
  ($a,b$) or ``NO.''\\[1mm]

  0 \hspace*{1em} \parbox[t]{0.85\linewidth}{
    find a minimum vertex \stsep{x}{y} $S$; 
    \\
    {\bf if} $|S|\le a$ {\bf then return } ($S,\emptyset$).
  }
  \\
  1 \hspace*{1em} $X = \emptyset$; $\quad$ $Y = \emptyset$; $\quad$ $Z
  = \emptyset$;
  \\
  2 \hspace*{1em} \parbox[t]{0.85\linewidth}{
    build a clique tree $\cT^F$ for $F$; 
    \\
    {\bf guess} a bag $K$ from the path of bags connecting $\cT^F(x)$
    and $\cT^F(y)$;
  }
  \\
  3 \hspace*{1em} {\bf enqueue}(${\cal Q}, K$);
  \\
  4 \hspace*{1em} {\bf while} ${\cal Q}\ne\emptyset$ {\bf do}
  \\
  4.1 \hspace*{2em} $K$ = {\bf dequeue}($\cal Q$);
  \\
  4.2 \hspace*{2em} {\bf if} $|K\setminus (X\cup Y\cup Z)|> a - |Z| +
  b - |E(F)\cap (X\times Y)| +1$ {\bf then return} ``NO'';
  \\
  4.3 \hspace*{2em} {\bf guess} a partition ($X_K, Y_K, Z_K$) of
  $K\setminus (X\cup Y\cup Z)$;
  \\
  4.4 \hspace*{2em} $X = X\cup X_K$; $\quad$ $Y = Y\cup Y_K$; $\quad$
  $Z = Z\cup Z_K$;
  \\
  4.5 \hspace*{2em} {\bf if} $a<|Z|$ or $b<|E(F)\cap (X\times Y)|$
  {\bf then return} ``NO'';
  \\
  4.6 \hspace*{2em} {\bf for each} bag $K'$ adjacent to $K$ that is
  not ``processed'' {\bf do}
  \\
  \hspace*{5em} {\bf if} $K'$ intersects both $X$ and $Y$ {\bf then}
  {\bf enqueue}(${\cal Q}, K'$);
  \\
  4.7 \hspace*{2em} mark $K$ ``processed'';
  \\
  5 \hspace*{1em} $V_S = Z$; $\quad$ $E_S = E(F) \cap (X\times Y)$.
  \\
  6 \hspace*{1em} {\bf if} $x$ and $y$ are disconnected in $F - V_S -
  E_S$ {\bf then return} ($V_S, E_S$);
  \\
  \hspace*{4em} {\bf else return} ``NO.''

\end{quote} \vspace*{-6mm} \strut} $$\boxit{\box4}$$
\vspace*{-9mm}
\caption{Algorithm finding mixed separators in chordal graphs.}
\label{fig:alg-mixed-separator}
\end{figure*}
\begin{proof}
  We use the algorithm described in
  Figure~\ref{fig:alg-mixed-separator}.  If the size of minimum
  \stsep{x}{y}s is no more than $a$, then step 0 will give a correct
  separator, and hence main part of the algorithm looks for a solution
  with $E_S\ne\emptyset$.  Let us explain the variables used in the
  algorithm and formally state its invariants.  The algorithm
  processes bags one by one, and maintains a partition ($X, Y, Z$) of
  vertices in all bags that have been processed.  The partition can be
  arbitrary if there exists no mixed \stsep{x}{y} of the designated
  size.  Otherwise the partition satisfies for some mixed \stsep{x}{y}
  ($V^*_S, E^*_S$) of the designated size that
  \begin{inparaitem}
  \item[(1)] $X$ and $Y$ are in the same components of $F -
    V^*_S - E^*_S$ as $x$ and $y$ respectively; and
  \item[(2)] $Z\subseteq V^*_S$.
  \end{inparaitem}
  The queue $\cal Q$ keeps all bags to be processed, and a bag is
  enqueued if it intersect both $X$ and $Y$.  A bag to be processed
  must be adjacent to a previously process bag, and since the queue
  starts from a single bag, at the end of the algorithm, all processed
  bags induce a connected subtree of $\cT^F$.

  The algorithm has no false positives.  Therefore, to verify its
  correctness, we show that each inclusion-wise minimal mixed
  \stsep{x}{y} ($V_S, E_S$) of size at most ($a,b$) can be found.
  We initialize $\cal Q$ by guessing a bag in the path connecting
  $\cT^F(x)$ and $\cT^F(y)$ that is disconnected by the deletion of
  ($V_S, E_S$); the existence of such a bag follows from previous
  discussion.  Main work of the algorithm is done in the loop of step
  4, each iteration of which processes a bag in $\cal Q$.  Let $K$ be
  the bag under processing.  By assumption, if a vertex $v\in
  K\setminus (X\cup Y\cup Z)$ is not in $V^*_S$, then it has to be
  incident to an edge in $E^*_S$.  Let $b' = |K\setminus (X\cup Y\cup
  Z)| - (a - |Z|)$; then at least $b'$ vertices of $K$ will remain in
  $F - V^*_S - E^*_S$, and any nontrivial partition of it has at least
  $b' - 1$ edges (when one side has precisely one vertex).  It cannot
  exceed $b - |E(F)\cap (X\times Y)|$; this justifies the exit
  condition 4.2.  Steps 4.3--4.5 are straightforward.  Step 4.6
  enqueues bags that have to be separated by the deletion of ($V^*_S,
  E^*_S$).

  It remains to verify that ($V_S, E_S$) constructed in step 5 is the
  objective mixed separator, i.e., $V^*_S =V_S = Z$ and $E^*_S = E_S =
  E(F)\cap (X\times Y)$.  Since we have shown that $Z\subseteq V^*_S$
  and $E(F)\cap (X\times Y)\subseteq E^*_S$, and by assumption, $x$
  (resp., $y$) remains connected to $X$ (resp., $Y$) in $F - V_S -
  E_S$, it suffices to show that $X$ and $Y$ are disconnected in $F -
  V_S - E_S$.  Suppose for contradiction that there is an induced path
  $P$ connecting $v_x\in X$ and $v_y\in Y$ in $F - V_S - E_S$.  Let
  $P$ be the path $u_1 \cdots u_p$ where $u_1 = v_x$ and $u_p = v_y$.
  Without loss of generality, assume that all internal vertices of $P$
  are disjoint from $X\cup Y$.  Let $l$ be the smallest index such
  that $1<l<p$ and $u_l\sim Y$.  We argue that $u_l\sim X$ as well.
  Otherwise, let $l'$ be the largest index such that $1<l'<l$ and
  $u_{l'}\sim X$.  It is easy to verify that in $F - V_S$, subgraphs
  induced by $X$, $Y$, and $X\cup Y$ are all connected.  Hence we can
  find an induced \stpath{u_{l'}}{u_{l}} with all internal vertices in
  $X\cup Y$; this path and $u_{l'} \cdots u_{l}$ make a hole, which is
  impossible as $F$ is chordal.  Let $v'_x\in X$ and $v'_y\in Y$ be
  neighbors of $u_l$.  Note that all bags handled in step 4 induce a
  connected subtree of $\cT^F$, and in particular, it intersects both
  $\cT^F(v'_x)$ and $\cT^F(v'_y)$.  If $v'_x\sim v'_y$, then there is
  a bag containing $\{u_l, v'_x, v'_y\}$.  Let us focus on bags that
  contain $v'_x$ and $v'_y$.  At least one of such bags is separated,
  and all of them are then enqueued in concession.  If $v'_x\not\sim
  v'_y$, then $u_l$ is in any \stsep{v'_x}{v'_y}, and at least one bag
  that contains $u_l$ is handled.  In both cases, $u_l$ has to be in
  $X\cup Y\cup Z$.  This gives a contradiction, and hence ($V_S, E_S$)
  must be a mixed \stsep{x}{y}.  This completes the proof of the
  correctness.

  We now analyze the runtime.  In step 2, there are at most $|V|$ bags
  in the path connecting $\cT^F(x)$ and $\cT^F(y)$, and thus the bag
  $K$ can be found in $O(|V|)$ time.  Note that this step is run only
  once.  The only step that takes exponential time is 4.3.  The set
  $K\setminus (X\cup Y\cup Z)$ has $3^{|K\setminus (X\cup Y\cup Z)|}$
  partitions, and after each execution of step 4.3, the budget
  decreases by at least $|K\setminus (X\cup Y\cup Z)| - 1$.  In total,
  this is upper bounded by $3^{a+b+1}$.  This completes the proof.
\end{proof}

We remark that the problem of finding a mixed separator of certain
size is fixed-parameter tractable even in general graphs: the
treewidth reduction technique of Marx et
al.~\cite{marx-13-treewidth-reduction} can be used after a simple
reduction (subdivide each edge, color the new vertices red and the
original vertices black, and find a separator with at most $k_1$ black
vertices and at most $k_2$ red vertices). However, the algorithm of
Lemma~\ref{lem:alg-mix-separator} for the special case of chordal
graphs is simpler and much more efficient.

The definition of mixed separator can be easily generalized to two
disjoint vertex sets---we may simply shrink each set into a single
vertex and then look for a mixed separator for these two new vertices.
Another interpretation of Lemma~\ref{lem:alg-mix-separator} is the following.
\begin{corollary}\label{col:alg-mix-separator}
  Let $X$ and $Y$ be a pair of nonadjacent and disjoint sets of
  vertices in a chordal graph $F$.
  For any nonnegative integer $a\le k_1$,  in
  time $3^{k_1+k_2+1}\cdot |V(F)|^{O(1)}$ we can find the minimum number $b$
  such that $b\le k_2$ and there is a mixed \stsep{X}{Y} of size
  ($a,b$) or assert that there is no mixed \stsep{X}{Y} of size
  ($a,k_2$).
\end{corollary}

\section{Proof of Theorem~\ref{thm:alg-compression}}
\label{sec:alg-chordal-editing}

We are now ready to put everything together and finish the analysis of
the algorithm.  We say that a chordal editing set is minimum if there
exists no chordal editing set with a smaller size.  Note that a
segment is contained in a unique path of $H-M$, which determines 
$V_1$ and $V_2$.

\begin{proof}[Proof of Theorem~\ref{thm:alg-compression}]
  Let ($V^*_-, E^*_-, E^*_+$) be a minimum chordal editing set of $G$
  of size no more than ($k_1, k_2, k_3$).  We start from a closer look
  at how it breaks $H$; by Theorem~\ref{lem:bound-segments}, we may
  assume that $H$ contains $O(k^3)$ segments.  There are three options
  for breaking $H$.  In the first case, $V^*_-$ contains some
  junction, or $E^*_-$ contains some edge of $H$ that is in $M\times
  V_0$. In this case, we can branch on including one of these vertices
  or edges into the solution; there are $O(k^3)$ of them.  Otherwise,
  we need to delete an internal vertex or edge from some segment.  Let
  $d=2k + 4$.  In the second case, we delete either
  \begin{inparaenum}[(1)]
  \item a vertex that is at distance at most $d$ (on the cycle) from a
    junction; or
  \item an edge
  whose both endpoints are at distance at most $d$ (on the cycle) from
  a junction.
  \end{inparaenum}
  In particular, this case must apply when we are breaking a segment
  of length at most $2d$.  If one of the two aforementioned cases is
  correct, then we can identify one vertex or edge of the solution by
  branching.  In total, there are $O(k^4)$ branches we need to try.

  Henceforth, we assume that none of these two cases holds.  We still
  have to delete at least one vertex or edge from $H$; this vertex or
  edge must belong to some segment $[v_s,v_t]$ with $t - s > 2d$.
  This is the third case, where we use $s' = s+ d$ and $t' = t - d$.
  Recall that any segment $[v_s,v_t]$ belongs to some maximal path $P$
  of $H - M$, on which $V_1$ and $V_2$ are well defined.  For any pair
  of indices $i,j$ with $s\le i<i+3\le j\le t$, we use $U_{[i,j]}$ to
  denote the union of the set of bags in the nonempty subtree of $\cT
  - \{K_{\last{i}}, K_{\first{j}}\}$ that contains $\{K_{\last{i}+1},
  \dots, K_{\first{j}-1}\}$, plus the two vertices $v_i$ and $v_j$.  Let
  $G_{[i,j]}$ be the subgraph induced by $U_{[i,j]}$.
  \begin{claim}\label{lem:solution-and-segments}
    There must be some segment $[v_{ s}, v_{t}]$ with $t -s > 2 d$
    such that vertices $v_{s'}$ and $v_{t'}$ are disconnected in $G_{[
      s, t]} - V^*_- - E^*_-$.
  \end{claim}
  \begin{proof}
    We prove by contradiction.  Consider first a segment $[v_s,v_t]$
    with $t - s > 2d$.  Suppose for contradiction that $v_{s'}$ and
    $v_{t'}$ are connected in $G_{[s,t]} - V^*_- - E^*_-$.  We can
    find an induced \stpath{v_{s'}}{v_{t'}} $P_{[s',t']}$ in
    $G_{[s,t]} - V^*_- - E^*_-$, which has to visit every bag $K_\ell$
    with $\last{s'}\le \ell\le \first{t'}$.  Appending to it
    $v_s\cdots v_{s'}$ and $v_s\cdots v_{s'}$, we get a
    \stpath{v_{s}}{v_{t}} $P_{[s,t]}$ in $G_{[s,t]} - V^*_- - E^*_-$.
    %
    From $P_{[s,t]}$ we can extract an induced \stpath{v_{s}}{v_{t}}
    $P'_{[s,t]}$ of $G_{[s,t]} - V^*_- - E^*_-$.  It is also a
    \stpath{v_{s}}{v_{t}} of $G_{[s,t]}$, where the distance between
    $v_s$ and $v_t$ is $t - s > 2d$, and thus the length of
    $P'_{[s,t]}$ is larger than $2d > 2 k_3 + 4$.  On the other hand,
    a segment $[v_s, v_t]$ of length at most $2d$ remains intact in $G
    - V^*_- - E^*_-$ by assumption, which can be used as the
    \stpath{v_{s}}{v_{t}}.

    We have then obtained for each segment $[v_s, v_t]$ of $H$ an
    induced \stpath{v_{s}}{v_{t}} $P'_{[s,t]}$ in $G - V^*_- - E^*_-$.
    Concatenating all these paths, as well as edges of $H$ in $M\times
    V(G)$, we get a closed walk $C$.  To verify that $C$ is a hole, it
    suffices to verify that the internal vertices of $P'_{[s,t]}$ is
    disjoint and nonadjacent to other parts of $C$.  On the one hand,
    no internal vertex of $P'_{[s,t]}$ is adjacent to $M\setminus A_M$
    by definition ($C$ is disjoint from $A$).  On the other hand, all
    internal vertices of $P'_{[s,t]}$ appear in the subtree that
    contains $K_{\last{s+4}}$ in $\cT - \{K_{\last{s+3}},
    K_{\first{t-3}}\}$, while no vertex in the \stpath{v_t}{v_s} in
    $C$ does.  This verifies that $C$ is a hole of $G - V^*_- -
    E^*_-$.  Since the length of $C$ is longer than $2 k_3 + 4$,  it cannot be made chordal by the addition of the at most $k_3$ edges of $E^*_+$. This contradiction proves the claim.
    \renewcommand{\qedsymbol}{$\lrcorner$}
  \end{proof}

  In other words, there is a segment $[v_s,v_t]$ such that ($V^*_-,
  E^*_-$) contains some inclusion-wise minimal mixed
  \stsep{\{v_s,\dots,v_{s'}\}}{\{v_{t'},\dots,v_t\}} ($V^*_S, E^*_S$)
  in $G_{[s,t]}$.  The resulting graph obtained by deleting ($V^*_S,
  E^*_S$)
  from $G_{[s,t]}$ is characterized by the following claim.
  \begin{claim}\label{claim-6}
    Let ($V_S, E_S$) be an inclusion-wise minimal mixed
    \stsep{v_{s'}}{v_{t'}} in $G_{[s',t']}$, and let $G' = G - V_S -
    E_S$.  Let $X$ be the component of $G' - (K_{\last{i}} \cup A)$
    for some $i$ with $s \le i\le s'$ that contains $v_{s'}$.  Then
    $X$ is simplicial in $G'$.
  \end{claim}
  \begin{proof}
    By definition, $N_{G'}(X) \subseteq K_{\last{i}} \cup A$ and is a
    clique in $G$; otherwise $v_{i+1}$ must be a junction of type (3),
    which is impossible.  Since ($V_S, E_S$) is inclusion-wise
    minimal, no edge in $E_S$ is induced by $N_{G'}[X]$.  In
    particular, $N_{G'}(X)$ induces the same subgraph in $G$ and $G'$,
    which is a clique.  It remains to show that $N_{G'}[X]$ induces a
    chordal subgraph of $G'$.  A vertex in $N_{G'}[X]$ is either in
    $V_2$, some branch, or $A$.  For every branch $C$ near to some
    vertex $v_i$ with $s< i< t$, $C\cap N_{G'}[X]$ is simplicial.  On
    the other hand, by definition of segments, $V_2\cap N_{G'}[X]$ is
    completely connected to $A$.  Therefore, $N_{G'}[X]$ induces a
    chordal subgraph in $G'$.  \renewcommand{\qedsymbol}{$\lrcorner$}
  \end{proof}

  A symmetric claim holds for the other side of the segment $[v_s,
  v_t]$.  That is, for any $i$ with $t' \le i\le t$, the component
  $X$ of $G' - (K_{\last{i}} \cup A)$ that contains $v_{t'}$ is
  simplicial in $G'$.  We now consider the subgraph obtained from $G$
  by deleting ($V^*_S, E^*_S$), i.e., $G' = G -V^*_S - E^*_S$.  Note
  that ($V^*_-\setminus V^*_S, E^*_-\setminus E^*_S, E^*_+$) is a
  minimum chordal editing set of $G'$.

  \begin{claim}\label{lem:break-segment-2}
    For any mixed \stsep{\{v_s,\dots,v_{s'}\}}{\{v_{t'},\dots,v_t\}}
    ($V^*_S, E^*_S$) of size at most ($|V^*_S|, |E^*_S|$) in
    $G_{[s,t]}$, substituting ($V_S, E_S$) for ($V^*_S, E^*_S$) in
    ($V^*_-, E^*_-, E^*_+$) gives another minimum editing set to $G$.
  \end{claim}
  \begin{proof}
    We first argue the existence of some vertex $v_{s''}$ with $s\le
    s''\le s'$ such that $E^*_-$ contains no edge induced by
    $K_{\last{s''}}$.  For each $s''$ with $s\le s''\le s'$, since
    $\last{s''}\ge \first{s''+1}$ and every vertex in them is adjacent
    to at most $3$ vertices of $H$ (Proposition~\ref{lem:at-most-3}),
    bags $K_{\last{s''}}$ and $K_{\last{s''+2}}$ are disjoint. In
    particular, an edge cannot be induced by both $K_{\last{s''}}$ and
    $K_{\last{s''+2}}$. Suppose that $E^*_-$ contains an edge induced by
    $K_{\last{s''}}$ for each $s''$ with $s\le s''< s'$, then we must
    have $|E_-|>(s'-s)/2\ge k_2$, which is impossible.  Likewise, we
    have some vertex $v_{t''}$ with $t'\le t''\le t$ such that $E^*_-$
    contains no edge induced by $K_{\first{t''}}$.  By
    Claim~\ref{claim-6}, it follows that every vertex of
    $U_{[s'',t'']}$ is in a simplicial set of $G-V^*_S-E^*_S$.  Since
    ($V^*_- \setminus V^*_S, E^*_- \setminus E^*_S, E^*_+$) is a
    minimum chordal editing set to $G - V^*_S - E^*_S$, we have by
    Lemma~\ref{lem:chordal-protrusion} that ($V^*_- \setminus V^*_S,
    E^*_- \setminus E^*_S, E^*_+$) does not edit any vertex of
    $U_{[s'',t'']}$.

    Suppose that there is a hole $C$ in the graph obtained by applying
    ($(V^*_- \setminus V^*_S) \cup V_S, (E^*_- \setminus E^*_S) \cup
    E_S, E^*_+$) to $G$.  By construction, $C$ contains a vertex of
    $U_{[{s'}, {t'}]}\subseteq U_{[s'',t'']}$. However, by
    Claim~\ref{claim-6}, every vertex of $U_{[s'',t'']}$ is in some
    simplicial set of $G-V_S-E_S$ and, as ($V^*_- \setminus V^*_S,
    E^*_- \setminus E^*_S, E^*_+$) does not edit $U_{[s'',t'']}$,
    every such vertex is in a simplical set after applying ($(V^*_-
    \setminus V^*_S) \cup V_S, (E^*_- \setminus E^*_S) \cup E_S,
    E^*_+$) to $G$. Thus no vertex of $U_{[s'',t'']}$ is on a hole, a
    contradiction.  \renewcommand{\qedsymbol}{$\lrcorner$}
  \end{proof}

  For any segment $[v_s, v_t]$, we can use
  Corollary~\ref{col:alg-mix-separator} to find all possible sizes of
  a minimum mixed \stsep{\{v_s,\dots,v_{s'}\}}{\{v_{t'},\dots,v_t\}}.
  There are at most $k_1$ of them.  By
  Claim~\ref{lem:break-segment-2}, one of them can be used to compose
  a minimum chordal editing set.
  In each iteration, we branch into $O(k^4)$ instances to break a
  hole, and in each branch decreases $k$ by at least $1$.  The runtime
  is thus $O(k)^{4k}\cdot n^{O(1)} = 2^{{O}(k \log{k})}\cdot
  n^{O(1)}$.  This completes the proof.
\end{proof}

\section{Concluding remarks}\label{sec:alg-disjoint-hc}
We have presented the first FPT algorithm for the general modification
problem to a graph class that has infinite number of obstructions.  It
is natural to ask for its parameterized complexity on other related
graph classes, especially for those classes on which every
single-operation version is already known to be FPT.  The most
interesting candidates include unit interval graphs and interval
graphs.  The fixed-parameter tractability of their completion versions
were shown by Kaplan et al.~\cite{kaplan-99-chordal-completion} and
Villanger et al.~\cite{villanger-09-interval-completion}; their vertex
deletion versions were shown by van 't Hof and Villanger
\cite{villanger-13-pivd} and Cao and Marx
\cite{cao-12-interval-deletion}.  A very recent result of
Cao~\cite{cao-14-almost-interval-recognition} complemented them by
showing that the edge deletion versions are FPT as well.

We would like to draw attention to the similarity between {\sc chordal
  deletion} and the classic {\sc feedback vertex set} problem, which
asks for the deletion of at most $k$ vertices to destroy all
\emph{cycles} in a graph, i.e., to make the graph a forest.  The
ostensible relation is that the forbidden induced subgraphs of forests
are precisely all holes and triangles.  But triangles can be easily
disposed of and its nonexistence significantly simplifies the graph
structure.  On the other hand, each component of a chordal graph can
be represented as a clique tree, which gives another way to be
correlate these two problems.

Recall that vertices with degree less than two are irrelevant for {\sc
  feedback vertex set}, while degree two vertices can also be
preprocessed, and thus it suffices to consider graphs with minimum
degree three.  Earlier algorithms for {\sc feedback vertex set} are
based on some variations of the upper bounds of Erd\H{o}s and P\'osa
\cite{erdos-62-number-disjoint-circuits} on the length of shortest
cycles in such a graph.  For {\sc chordal vertex deletion}, our
algorithm can be also interpreted in this way.  First of all, a
simplicial vertex participates in no holes, and thus can be removed
safely.

\begin{kernelrule}\label{rule:free-vertex}
  Remove all simplicial vertices.
\end{kernelrule}
Note that a simplicial vertex corresponds to a leaf in the clique
tree, Reduction~\ref{rule:free-vertex} can be viewed as a
generalization of the disposal of degree-1 vertices for
\textsc{feedback vertex set}.  For \textsc{feedback vertex set}, we
``smoothen'' a degree-2 vertex by removing it and adding a new edge to
connect its two neighbors.  This operation shortens all cycles through
this vertex and result in an equivalent instance.  To have a similar
reduction rule, we need an explicit clique tree,\footnote{This can be
  surely extended to some local clique tree structure, and we use
  clique tree here for simplicity.} so we consider the compression
problem, which, given a hole cover $M$, asks for another hole cover
$M'$ disjoint from $M$.  The following reduction rule will only be
used after Reduction \ref{rule:free-vertex} is not applicable, then no
vertex inside a segment can have a branch.  Let $S_\ell$ denote the
separator $K_\ell\cap K_{\ell+1}$ in the clique tree.

\begin{kernelrule}\label{rule:degree-2-node}
  Let $[v_s,v_t]$ be a segment and $|S_{\imath}| = \min_{\last{s} \leq
    i < \first{t}} |S_{i}|$.  If there exists $S_\ell$ such that
  $S_\ell$ is disjoint from $K_{\last{s}} \cup K_{\first{t}}$ and
  there exists $v\in S_\ell\setminus S_{\imath}$, then remove $v$ and
  insert edges to make $N(v)$ a clique.
\end{kernelrule}

After both reductions are exhaustively applied, we can use an argument
similar as Theorem~\ref{lem:bound-segments} to show that either the
length of a shortest hole is $O(k^4)$ or there is no solution.
However, unlike {\sc feedback vertex set}, Reductions
\ref{rule:free-vertex} and \ref{rule:degree-2-node} do not directly
imply a polynomial kernel for {\sc chordal vertex deletion}.
Therefore, we leave it open the existence of polynomial kernels for
the {\sc chordal vertex deletion} problem and its compression
variation.

\end{document}